\documentclass{article}

\usepackage{amsmath}
\usepackage{amsfonts,amssymb}
\usepackage{eucal}
\usepackage{amsthm}
\usepackage{graphicx}
\usepackage{color} 

\newcommand{\simgt}{\lower.5ex\hbox{$\; \buildrel > \over \sim \;$}}
\newcommand{\simlt}{\lower.5ex\hbox{$\; \buildrel < \over \sim \;$}}

\newcommand{\ord}{\text{\rm ord}}
\newcommand{\Var}{\text{\rm Var}}

\newcommand{\Mod}{\text{\rm mod }}
\newcommand{\lcm}{\text{\rm lcm}}

\newtheorem{thm}{Theorem}

\newtheorem{defn}[thm]{Definition}

\begin{document}
%
\title{Double Counting in $2^t$-ary RSA Precomputation Reveals the Secret Exponent}
%
%
%

\author{Masahiro~Kaminaga, 
        Hideki~Yoshikawa,
        and~Toshinori~Suzuki,
\thanks{M. Kaminaga, H. Yoshikawa, and T. Suzuki are with the Department 
of Electrical Engineering and Information Technology, Tohoku Gakuin University, 
13-1, Chuo-1, Tagajo, 985-8537, Japan e-mail: (kaminaga@mail.tohoku-gakuin.ac.jp).}
\thanks{
}}

%
%

\markboth{
}%
{Kaminaga \MakeLowercase{\textit{et al.}}: Double Counting in $2^t$-ary RSA Precomputation Reveals the Secret Exponent}
%



\maketitle

\begin{abstract} 
A new fault attack, double counting attack (DCA), 
on the precomputation of $2^t$-ary modular exponentiation for 
a classical RSA digital signature (i.e., RSA without the Chinese remainder theorem) 
is proposed. 
The $2^t$-ary method is the most popular and widely used 
algorithm to speed up the RSA signature process. 
Developers can realize the fastest signature process by choosing optimum $t$. 
For example, $t=6$ is optimum for a 1536-bit classical RSA implementation. 
The $2^t$-ary method requires precomputation to generate 
small exponentials of message. 
Conventional fault attack research has paid little attention to 
precomputation, even though precomputation could be a target of a fault attack. 
The proposed DCA induces faults in precomputation by using instruction skip technique, 
which is equivalent to replacing an instruction with a no operation in assembly language. 
This paper also presents a useful ``position checker" tool to determine the 
position of the $2^t$-ary coefficients of the secret exponent from signatures based on 
faulted precomputations. 
The DCA is demonstrated to be an effective attack method for some widely used parameters. 
DCA can reconstruct an entire secret exponent using the position checker with $63(=2^6-1)$ faulted signatures 
in a short time for a 1536-bit RSA implementation using the $2^6$-ary method. 
The DCA process can be accelerated for a small public exponent (e.g., 65537). 
The the best of our knowledge, the proposed 
DCA is the first fault attack against classical RSA precomputation. 
\end{abstract}


%

\section{Introduction}\label{intro}
Since the publication of Boneh, DeMillo, and Lipton's landmark paper\cite{Bellcore}, 
differential fault analysis (DFA) 
has been an active area of cryptography. 
DFA is a technique to extract secret information from
cryptographic device such as a smart card by provoking a computational fault. 
Faults are often caused by abnormal voltage, and clock signals\cite{LowCost}, 
or an optical flush \cite{Anderson}. 
Bar et al. \cite{Bar} and Osward and F.-X. Standaert (Chapter 1 in \cite{Joye-Tunstall-Book}) 
are good guides to practical application of fault attacks. 
Many researchers are developing DFA techniques for 
RSA signature schemes.
RSA is the most popular digital signature scheme. 
For signing a message $m$, 
the signer computes the signature $S=\mu(m)^d~\Mod N$ using an 
encoding function $\mu$, where $N=pq$ with distinct primes $p$ and $q$. 
In this paper, we consider a deterministic encoding function such as 
PKCS\#1 v.1.5 \cite{PKCS}, and a full domain hash scheme. 
To verify the signature $S$, the verifier checks the validity of 
the received signature $(S,m)$ with the public exponent $e$ 
to determine if $S^e\equiv \mu(m) (\Mod N)$ holds. 
There are two types of RSA signature implementations, 
classical RSA, i.e., RSA without the Chinese remainder theorem (CRT), 
and RSA-CRT, i.e., RSA with CRT. 
Traditionally, CRT has been used to speed up 
the RSA signing process; 
however, there are also many effective DFA techniques. 
Lenstra \cite{Lenstra} showed that the public modulus $N$ can be factored 
with only one signature/faulted signature pair.

On the other hand, in classical RSA implementations, obtaining the entire 
secret exponent and factorization of the public modulus with a few faulted signatures 
is more difficult. 
In many papers (e.g. Boneh-DeMillo-Lipton \cite{Bellcore}, Bao et al. \cite{Bao},
 and Yen-Joye \cite{SafeError}), 
the goal of the attacker is to gradually reconstruct the secret exponent 
rather than factor the public modulus. 
A unique exception is Seifert's attack \cite{Seifert}, which is 
based on the assumption that an attacker can induce faults as the device moves 
$N$ data from memory. Under this assumption, 
the attacker can create a new faulted modulus $\hat{N}$. 
Seifert pointed out that if $\hat{N}$ is prime, the attacker can compute 
a new secret exponent. 
Muir et al. \cite{Muir} simplified Seifert's attack and showed 
that the simplified version works under a relaxed condition, i.e., 
$\hat{N}$ can be factored. 
Under this condition, the attacker can obtain 
a new secret exponent $e^{-1}~\Mod \varphi(\hat{N})$, 
where, $\varphi$ is Euler's totient function.
Their attack methods are quite sophisticated; however, successful attack 
is entirely dependent on whether $\hat{N}$ can be factored.
In contrast, our attack method is not dependent on such a condition.

This paper proposes a new and effective fault attack, double counting attack (DCA)
targeting at precomputation process in $2^t$-ary modular exponentiation. 
Conventional fault attack reseach has paid little attention to precomputation in the 
$2^t$-ary method, which is an algorithm used to preform efficient calculations of modular exponentiation.
The proposed DCA is built on an instruction skip technique.
The instruction skip is equivalent to replacing an instruction with a no operation 
in assembly language. 
Let $\mathbb{Z}[a,b]$ be a set of integers in the interval $[a,b]$. 
We introduce a ``position checker" tool to determine 
the set ${\cal P}_{n,t,\ell}=\{j\in \mathbb{Z}[0,\lceil n/t\rceil-1]:d[t,j]=\ell\}$
for the $2^t$-ary representation of the $n$-bit secret exponent $d=\sum_{j=0}^{\lceil n/t\rceil -1}d[t,j](2^t)^j$. 
${\cal P}_{n,t,\ell}(\ell\in\mathbb{Z}[0,2^t-1])$ 
determines the entire secret exponent $d$ for typical choices of $n$ and $t$.
For example, DCA can reconstruct the entire secret exponent $d$ from 63 faulted signatures 
in a short time for $n=1536$ and $t=6$, 
which is the fastest parameter for $n=1536$.
The DCA process can be accelerated for small public exponents such as 65537.
To the best of our knowledge, the 
DCA is the first fault attack method targeting at the precomputation of the $2^t$-ary method. 

The remainder of this paper is organized as follows. 
Section II presents basic facts about the target implementation of DCA, 
the instruction skip technique, and a naive attack using the instruction skip.
Section III presents the position checker and describes the general principle 
and procedures of the DCA. 
Simulated attack results are presented in Section IV. 
An effective attack method for a small public exponent is described in Section V.
Several software countermeasures against DCA are presented in Section VI, 
and conclusions are presented in Section VII.

\section{Preliminaries}

\subsection{Modular Multiplication Coprocessor}
Coprocessors for smart card microcontrollers have been specifically designed 
to perform efficient calculations 
of public-key algorithms, such as modular multiplication 
in an RSA cryptosystem \cite{SmartcardHandbook}.
The RSA digital signature $S=M^d~\Mod N$ is performed using a 
modular multiplication ``$AB ~\Mod N$" coprocessor.
We assume that this coprocessor has three registers, $A$, $B$, and $N$. 
The coprocessor computes $AB ~\Mod N$ 
and writes the result to register $A$. 
Here, we assume this coprocessor has two modes, SQUARE [$A\leftarrow A^2 ~\Mod N$] 
and MULTIPLY [$A\leftarrow AB ~\Mod N$] (denoted ``SQARE enable" and ``MULTIPLY enable," respectively).
The coprocessor executes once when the enable bit is set to one by a specific instruction, 
such as a ``move" instruction. 
Some microcontrollers with coprocessors have a similar function. 
Since the microcontroller consumes a significant amount of power \cite{Joye-Tunstall-Book}, 
the attacker can easily distinguish the timing of modular multiplication 
from other instructions by performing a power consumption trace \cite{DPA}. 
This fact means that determining the timing of injected fault for 
skipping a conditional/unconditional 
branch instruction or increment/decrement instruction
is easier than it is for other DFA techniques. 
DCA is applicable to a microchip without a coprocessor 
if the timing of each $AB~\Mod N$ subroutine can be distinguished by the attacker. 
Here, we describe the attack for a case by using an $AB~\Mod N$ coprocessor; 
however, DCA also works for a modular multiplication coprocessor 
using the Montgomery method \cite{Mont}, i.e., 
compute $ABR^{-1} ~\Mod N(R=2^n)$ rather than $AB ~\Mod N$, 
where $N$ is $n$-bits long. 
In the Montgomery coprocessor, internal data are formatted as $\mbox{\text Mont}(A)=AR~\Mod N$.
\subsection{The $2^t$-ary method}
For classical RSA implementations, 
the $2^t$-ary exponentiation method is used to speed up the signing process. 
The $2^t$-ary method generalizes binary exponentiation and is based on $2^t$-ary 
representation of the exponent, which is expressed as follows: 
$$
d=\sum_{j=0}^{\lceil n/t\rceil-1}d[t,j](2^t)^j,  
$$
where $d[t,j]\in \mathbb{Z}[0,2^t-1]$. 
\begin{defn} The $2^t$-ary table ${\cal T}_{2^t}$ is defined as a set consisting of
$$
M^i~\Mod N \quad\mbox{for every $i\in \mathbb{Z}[1,2^t-1]$}.
$$
\end{defn}
The classical $2^t$-ary RSA implementation~\cite{HandbookofAppliedCryptography} is given in LIST I.
In the precomputation in List I, 
the microcontroller computes all values of ${\cal T}_{2^t}$ and stores them in RAM. 
The $2^t$-ary modular exponentiation uses the values in the table ${\cal T}_{2^t}$ 
after precomputation.
We discuss some stochastic properties of this precomputation in Section \ref{Stoc}. 
Here, we describe the value $t$ minimizing the execution time of $2^t$-ary modular exponentiation.

\begin{center}
{\small LIST I} \\
{\small Classical $2^t$-ary RSA implementation}
\end{center}
{\footnotesize
\begin{verbatim}
Input: M, n, t, d, N
Output: S = M^d mod N
# precomputation
Compute M^i mod N(i=1,...,2^t-1) 
        and store them in RAM
j = ceiling(n/t)-1
MOVE A <- 1
Loop:
  If j < 0 then  
    Break and Return A
  SQUARE enable (t times)
    [A <- A*A mod N: t times]
  If d[t,j] is not equal to 0 then
    MOVE M^d[t,j] mod N to register B
    MULTIPLY enable ---(*)
    [A <- A*B mod N]
  decrement j 
  goto Loop
\end{verbatim}
}

It is natural to assume that the average execution time of $A^2 ~\Mod N$ is equal to 
that of the $AB ~\Mod N$ coprocessor.
Then, the total execution time of the modular exponentiation using the $2^t$-ary method is proportional to the 
number of executions of the $AB ~\Mod N$ coprocessor. 
Since the number of MULTIPLY instructions required to generate table ${\cal T}_{2^t}$ 
is $2^t-1$, 
the number of SQUARE instructions is always $n$, and the average number of MULTIPLY 
instructions is $\frac{n}{t}\left(1-\frac{1}{2^t}\right)$. 
Therefore, we can estimate the average execution time $\tau(n,t)$ from the 
average execution time of MULTIPLY as follows:
$$
\tau(n,t) = (2^t-1)+n+\frac{n}{t}\left(1-\frac{1}{2^t}\right).
$$
Here, $\tau(n,t)$ is a convex function in $t$. 
We show $\tau(n,t)/n$ for $n=1024, 1536$, and $2048$ in Fig.\ref{windowsize}.
\begin{figure}[t]
 \begin{center}
  \includegraphics[width=80mm]{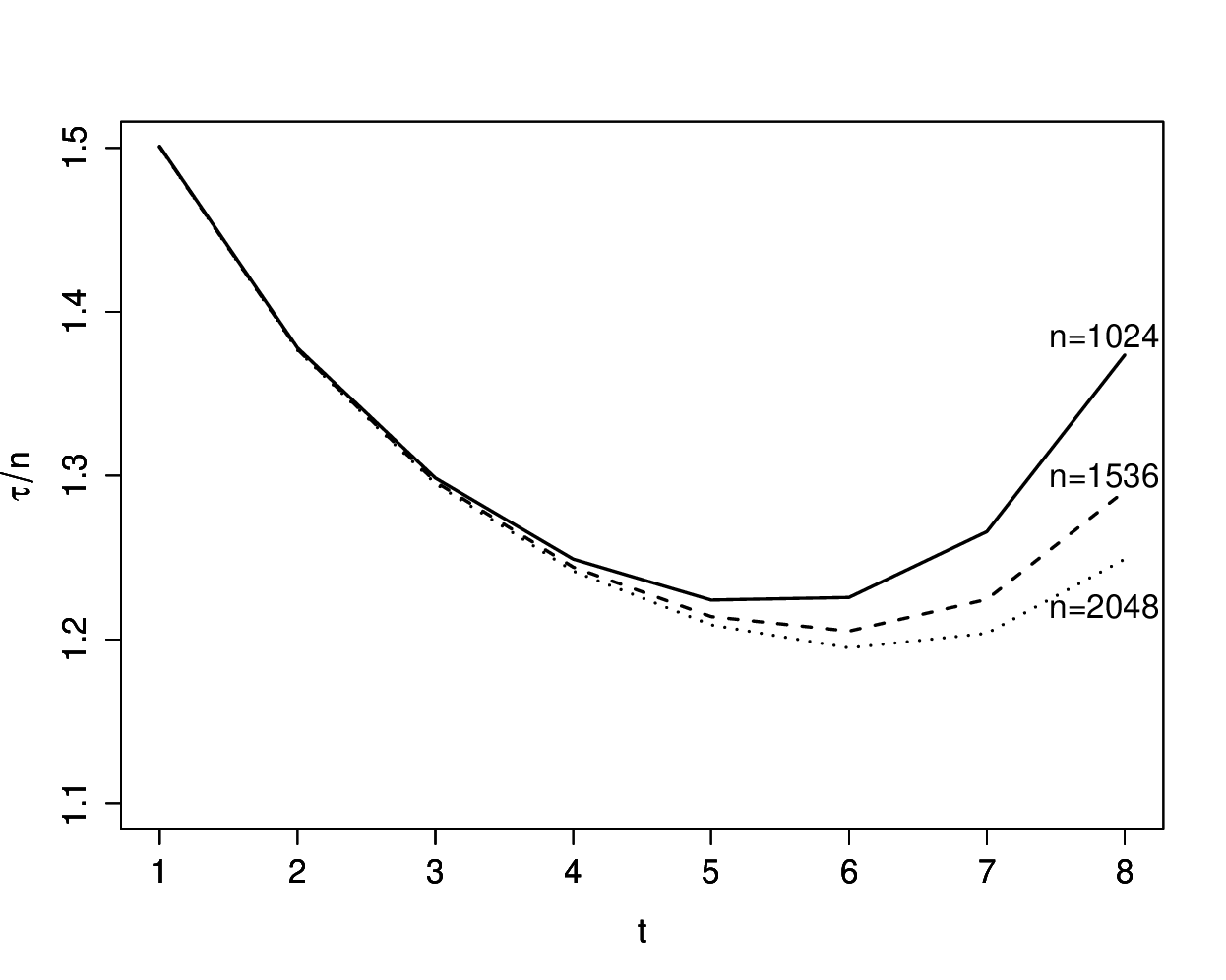}
 \end{center}
 \caption{$\tau(n,t)/n$ for $n=1024$, $1536$, and $2048$}
 \label{windowsize}
\end{figure}
\begin{table}[t]
\begin{center}
\scriptsize
\caption{$\tau(n,t)$ for $n=1024$, $1536$, and $2048$}\label{averagetime}
  \begin{tabular}{|c||c|c|c|c|c|c|c|c|}
  \hline
$t$  &    3     &     4    &     5    &     6    &     7    &     8    \\
\hline
\hline
$\tau(1024,t)$ & 1329.7 & 1279.0 & {\bf 1253.4} &  1255.0 & 1296.1  & 1406.5  \\
\hline
$\tau(1536,t)$ & 1991.0 & 1911.0 & 1864.6 & {\bf 1851.0} & 1880.7 & 1982.3 \\ 
\hline
$\tau(2048,t)$ & 2652.3 & 2543.0 & 2475.8 & {\bf 2447.0} & 2465.3 & 2558.0 \\ 
\hline
  \end{tabular}
\end{center}
\end{table}
We also show precise values of $\tau(n,t)$ for $n=1024, 1536$, and $2048$ in Table \ref{averagetime}.
$\tau(1024,t)$, $\tau(1536,t)$, and $\tau(2048,t)$ 
take a minimum at $t=5, 6,$ and $6$, respectively. 
Note that there is very little difference between $\tau(1024,5)$ and $\tau(1024,6)$. 
The actual execution time depends on the length of the exponent $n$ and $t$, which 
is restricted by hardware specifications such as the amount of RAM available. 
Most developers of RSA primitives are likely to choose an optimum or near-optimum value for $t$.

\subsection{Instruction Skip Technique}
Most DFA papers assume that injected faults affect several bits of the internal data. 
However, several researches have investigated DFA using an instruction skip, or a bypass operation 
\cite{RoundReduction}, \cite{park}, \cite{yoshikawa1}. 
The instruction skip does not affect the registers, internal memory, and calculation process.
Successful instruction skip attacks have been reported for the 
PIC16F877~\cite{RoundReduction}, ATmega 128~\cite{park}, 
and ATmega 168~\cite{yoshikawa1} microcontrollers. 
Choukri-Tunstall~\cite{RoundReduction} and Park et al.~\cite{park} showed that 
an entire Advanced Encryption Standard secret key could be 
reconstructed by skipping a branch instruction used to 
check the number of rounds. 
Yoshikawa et al.~\cite{yoshikawa1} described a method 
to reconstruct the entire secret key of 
several ciphers of a generalized Feistel network, such as CLEFIA, 
by skipping the increment or decrement 
instruction used to count the number of rounds.

\subsection{Naive Attack Using an Instruction Skip}\label{Naive}
The simplest attack on RSA using an instruction skip attacks 
the $j$-th MULTIPLY enable command (*) for every $j\in \mathbb{Z}[0,\lceil n/t\rceil-1]$.
When the $j$-th MULTIPLY enable command is skipped, the faulted result can be represented by
$$
\hat{S}_j = M^{d-d[t,j](2^t)^j}\Mod N.
$$
Therefore, we obtain 
$$
S\hat{S}_j^{-1}\equiv M^{d[t,j](2^t)^j}(\Mod N).
$$
Then, the attacker determines the $j$-th coefficient $d[t,j]$ 
by comparing $S\hat{S}_j^{-1}\Mod N$ with values of 
$M^{\ell(2^t)^j}~\Mod N$ for $\ell\in \mathbb{Z}[0,2^t-1]$. 
Essentially, the same attack on a binary RSA implementation was 
proposed by Kaminaga et al.~\cite{KaminagaDFA}.
This attack requires $\lceil n/t\rceil-1$ fault injections.
For example, the attacker needs 255 fault injections for the case $n=1536$ and $t=6$.
From the attacker's point of view, fewer fault injections are desirable because 
an actual attack generally requires careful timing of injection. 
DCA requires fewer faults than the given naive example.

\section{Proposed Attack}

\subsection{Position Checker}
The central idea of DCA is to determine the position set 
${\cal P}_{n,t,\ell}=\{j\in \mathbb{Z}[0,\lceil n/t\rceil-1]: d[t,j] = \ell\}$ 
using the position checker. 
Here, we define the position exponent and position checker. 

\begin{defn}
For any subset ${\cal A}$ of $\mathbb{Z}[0,\lceil n/t\rceil-1]$, we define 
$$
d[t,{\cal A}]=\sum_{j\in {\cal A}}(2^t)^j.
$$
Here, $d[t,{\cal A}]=0$ whenever ${\cal A}=\emptyset$.
The position exponent for coefficient $\ell$ is defined by $d[t,{\cal P}_{n,t,\ell}]$.
\end{defn}
An example of position exponents for $t=2$ is represented in Fig. \ref{position_cheker}.
The exponent $d$ can be recovered from $\{d[t,{\cal P}_{n,t,\ell}], \ell\in \mathbb{Z}[0,2^t-1]\}$ as
$$
d = \sum_{\ell=0}^{2^t-1}\ell d[t,{\cal P}_{n,t,\ell}].
$$

\begin{defn}
Let $\mathbb{Z}_N^*=(\mathbb{Z}/N\mathbb{Z})^*$ be the reduced residue system mod $N$.
For every $\ell\in \mathbb{Z}[0,2^t-1]$, 
the position checker for coefficient $\ell$ is defined by 
the map from $\mathbb{Z}_N^*$ to itself:
$$
C_{t,\ell}(x) = x^{d[t,{\cal P}_{n,t,\ell}]} ~\Mod N.
$$

\end{defn}

\begin{figure}[htbp]
 \begin{center}
  \includegraphics[width=80mm]{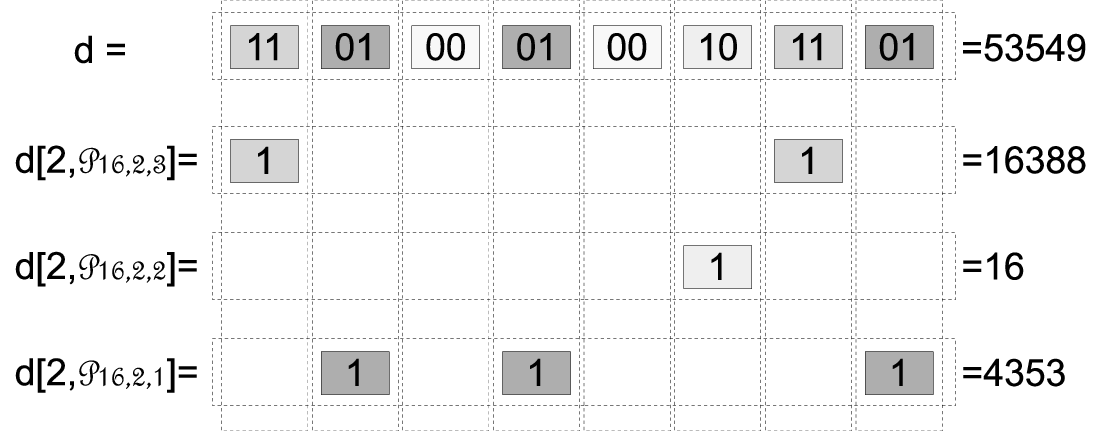}
 \end{center}
 \caption{Position exponent for $t=2$. Blank cells denote zero.}
 \label{position_cheker}
\end{figure}

\subsection{Stochastic Properties of the $2^t$-ary Table}\label{Stoc}
In this section, we describe the stochastic properties of table ${\cal T}_{2^t}$.
The target process of DCA is the precomputation.
We observe that the number of $2^t$-ary coefficients in $d$ 
is equal to $\lceil n/t\rceil$, 
and $M^j~\Mod N$ for every $j\in \mathbb{Z}[1,2^t-1]$ is stored in RAM.
Let $X_{n,t,\ell}=\sharp{\cal P}_{n,t,\ell}$, 
where the cardinality of set ${\cal A}$ is denoted $\sharp{\cal A}$. 
$2^t$-tuples of random variables $(X_{n,t,0}, \cdots, X_{n,t,2^t-1})$ 
obey multinomial distribution with probabilities 
$\mathbb{P}(X_{n,t,0}=m)=2^{-t}, \cdots, \mathbb{P}(X_{n,t,2^t-1}=m)=2^{-t}$ 
for each $m\in\mathbb{Z}[0,\lceil n/t\rceil-1]$.
Here, we only consider the $\ell$-independent stochastic properties of $X_{n,t,\ell}$; 
therefore, we omit $\ell$ of $X_{n,t,\ell}$. 
Thus, the expectation and variance of $X_{n,t}$ are given by 
$$
\mathbb{E}[X_{n,t}] = \frac{n}{t2^t},\quad \Var[X_{n,t}]=\frac{n}{t2^t}\left(1-\frac{1}{2^t}\right),
$$
respectively.
Its probability histogram is of approximately normal distribution for small $t$,  
whereas $\mathbb{E}[X_{n,t}]$ is large.

On the other hand, the histogram is strongly skewed to the left for relatively large $t$ 
and can be approximated by Poisson distribution.
The histogram of $X_{n,t}$ for $t=5$ and $t=6$ are shown in Fig. \ref{dist}.

\begin{figure}[htbp]
 \begin{center}
  \includegraphics[width=80mm]{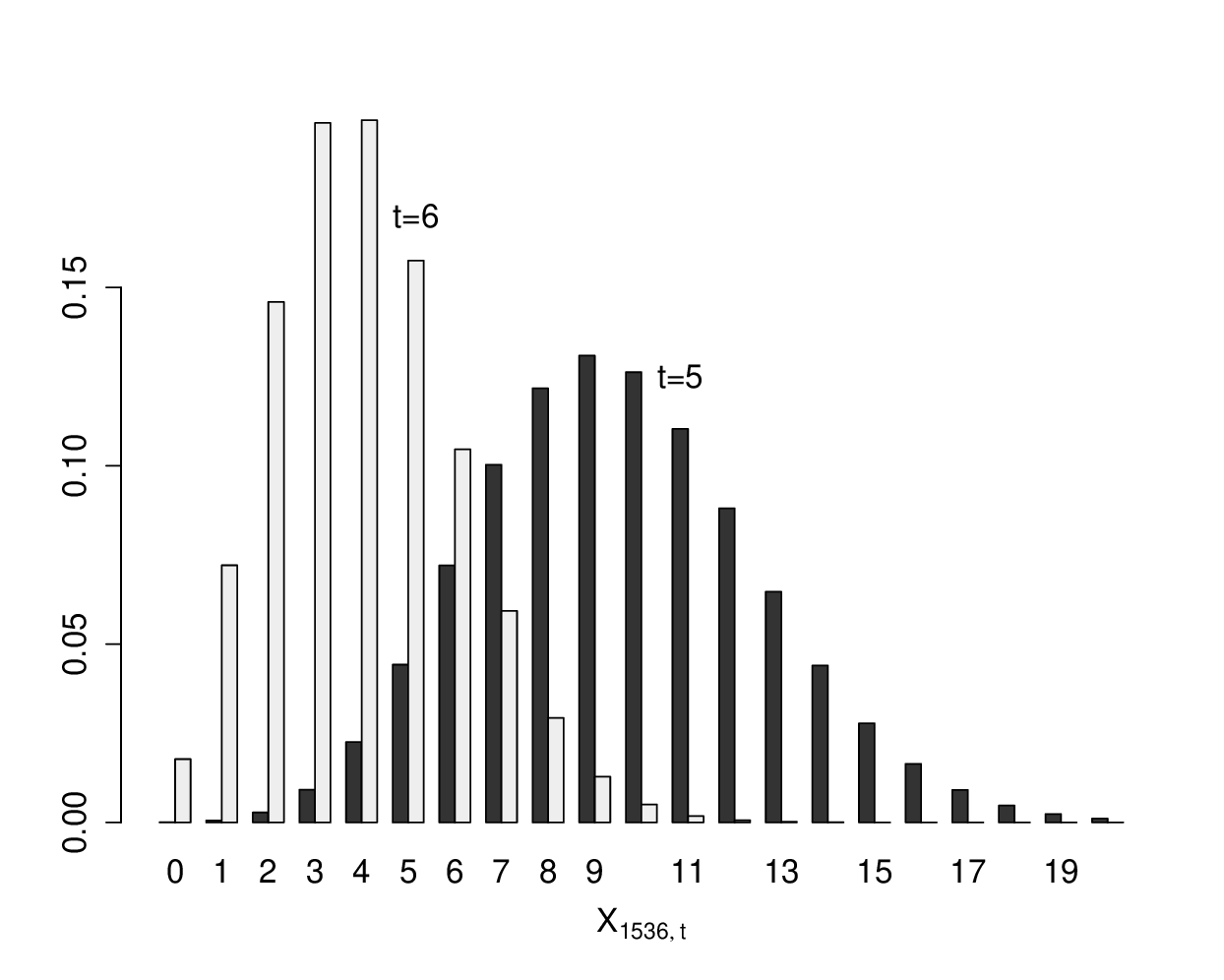}
 \end{center}
 \caption{Distribution of $X_{1536,t}$ for $t=5$ and $t=6$}
 \label{dist}
\end{figure}

Consider the relationship between the binary logarithm of the number of positions 
given by
$$
V_{k,t}=\log_2\left(\sum_{z=0}^k
{\lceil n/t\rceil \choose z}
\right)
$$
and the cumulative probability $\mathbb{P}(X_{n,t}\leq k)$. 
Here, $V_k$ is the bit size of the attacker's target search space.
We plot points $\{(V_{k,t},\mathbb{P}(X_{1536,t}\leq k)):k\geq 0\}$ 
for $t=4, 5$, and $6$ in Fig. \ref{positions}.
\begin{figure}[t]
 \begin{center}
  \includegraphics[width=80mm]{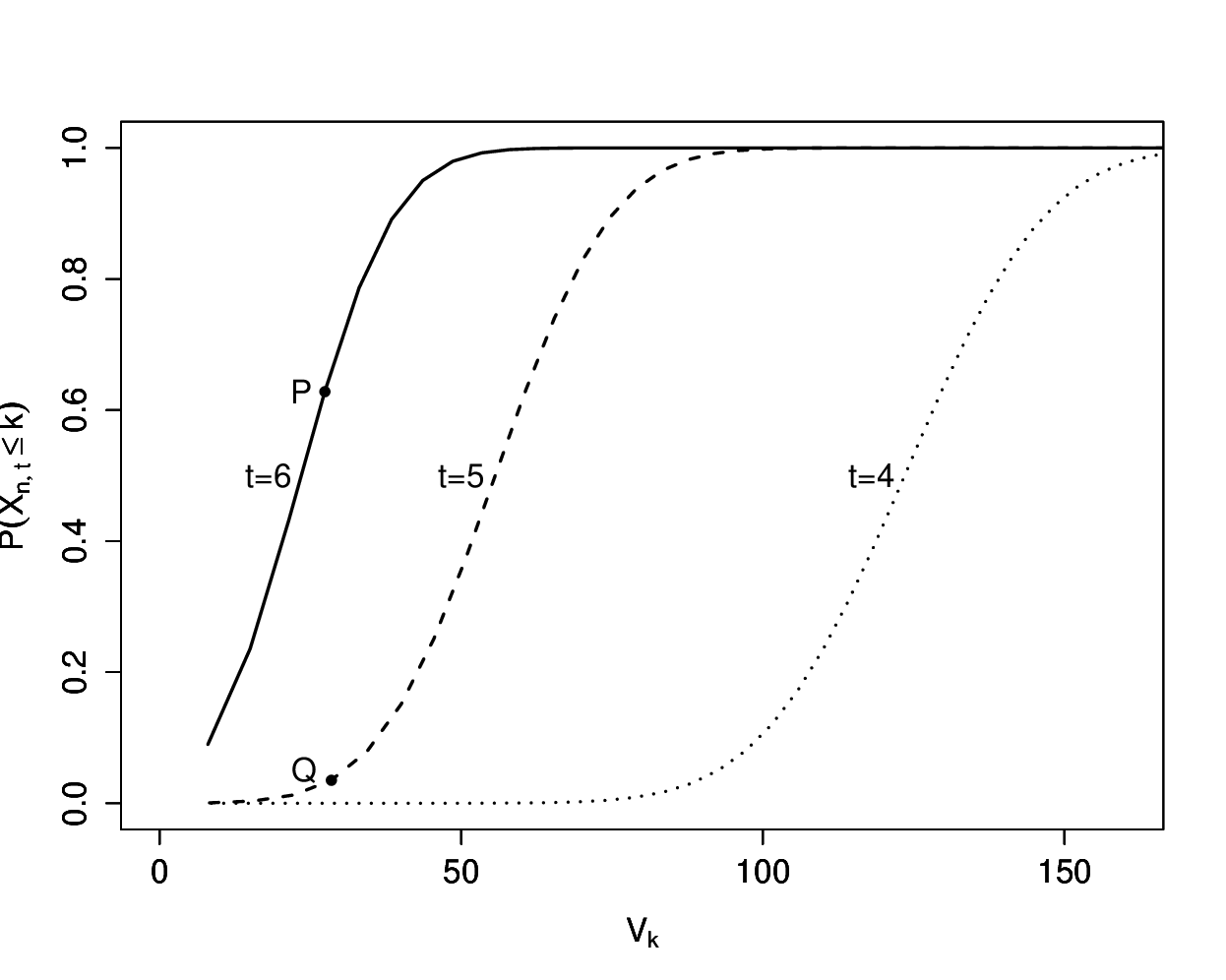}
 \end{center}
 \caption{Relation between $V_k$ and $\mathbb{P}(X_{n,t}\leq k)$ for $n=1536$}
 \label{positions}
\end{figure}
For example, for $t=6$, 
$$
(V_{4,6},\mathbb{P}(X_{1536,6}\leq 4)) = P(27.4,0.629).
$$
This $V_{4,6}$ is sufficiently small for an attacker to search all 
combinations of position by brute force.
The probability $\mathbb{P}(X_{n,t}\leq 4)$ is quite sensitive to $t$. 
As a matter of fact, for $t=5$, 
$$
(V_{4,5},\mathbb{P}(X_{1536,5}\leq 4)) = Q(28.5,0.0351).
$$
This means that the bit size of search space $V_{4,5}$ grows twice of $V_{4,6}$. 
However, obtaining coefficients mere $3.5\%$.

\subsection{Generating Position Checker from Faulted Signatures}\label{genposck}
Here, we describe how to generate the position checker from faulted signatures.
We first call the standard precomputation algorithm (List II).
The table ${\cal T}_{2^t}$ generated according to 
the following recursive relation.
$$
M^j=M^{j-1}M ~\Mod N \mbox{ for $j\in \mathbb{Z}[1,2^t-1]$}
$$

If the MULTIPLY enable command (*) at $j=k$ is skipped in List II, 
the microcontroller generates a faulted table given by
$$
M, M^2,\cdots ,M^{k-1},M^{k-1}\cdots , M^{2^t-2}.
$$
In this faulted table, $M^{k-1}$ appears twice, and $k-1$ is double counted.
The $j$-th value in this faulted table can be represented by 
$M^j~\Mod N$ for $j \in \mathbb{Z}[1,k-1]$, 
and $M^{j-1}~\Mod N$ for $j \in \mathbb{Z}[k,2^t-1]$.

\begin{center}
{\small LIST II} \\
{\small Precomputation for $2^t$-ary RSA implementation}
\end{center}
{\footnotesize
\begin{verbatim}
Input:M,t,N
Output: M[j] = M^j mod N, j=1,2,...,2^t-1
A <- 1, B <- M, j = 1
Loop:
  If j >= 2^t then 
    break and return 1, M mod N,..., M^(2^t-1) mod N
  MULTIPLY enable  ---(*)
       [A <- AB mod N]
  Move A -> RAM as M^j mod N
  increment j
goto Loop
\end{verbatim}
}

Note that the size of the table does not vary by the attack.
We denote such a faulted table as $\hat{{\cal T}}_{2^t}(k)$.
The faulted signature $\hat{S}_{t,k}$ generated from $\hat{{\cal T}}_{2^t}(k)$ is given by
$$
\hat{S}_{t,k} = M^{d-d[t,\cup_{k\leq \ell\leq 2^t-1}{\cal P}_{n,t,\ell}]}~\Mod N.
$$
Therefore, we obtain the following position checker:
\begin{eqnarray*}
\hat{S}_{t,k+1}\hat{S}_{t,k}^{-1} &\equiv& M^{d[t,{\cal P}_{n,t,k}]} \nonumber \\
&\equiv& C_{t,k}(M)(\Mod N)\mbox{ for $k\in \mathbb{Z}[1,2^t-2]$.}
\end{eqnarray*}
Note that $\hat{S}_{t,2^t-1}$ is equal to the correct signature $S$, and we cannot 
determine $C_{t,0}$. 
Therefore, the set ${\cal P}_{n,t,0}$ is determined as the complement set of 
$\cup_{\ell\in\mathbb{Z}[1,2^t-1]}{\cal P}_{n,t,\ell}$.

\subsection{Attack Procedure}\label{attackproc}
The attacker determines the position set ${\cal P}_{n,t,k}$ using the position checker.
The attacker knows that there is one correct position for $k$; 
all other positions are incorrect.
The goal is to identify the correct position. 
The following steps provide an example of the DCA process.
Note that processes 2 and 3 are completely off-line attacks.
\begin{enumerate}
\item Collect $2^t-1$ faulted signatures $\hat{S}_k$ 
from the faulted table ${\cal \hat{T}}_{2^t}(k)$ for every $k\in \mathbb{Z}[1,2^t-1]$.
\item The attacker computes position checkers $C_{t,k}(M)$ from the relation 
$\hat{S}_{t,k+1}\hat{S}_{t,k}^{-1}\equiv C_{t,k}(M)(\Mod N)$ for every $k\in \mathbb{Z}[1,2^t-2]$.
\item The attacker chooses a position set ${\cal \tilde{P}}_{n,t,k}$ for $k$ and computes 
the following:
$$
\tilde{C}_{t,k}(M) = M^{d[t,{\cal \tilde{P}}_{n,t,k}]}~\Mod N.
$$
The attacker checks whether the value $\tilde{C}_{t,k}(M)$ is equal to the 
value of the position checker $C_{t,k}(M)$. 
If $\tilde{C}_{t,k}(M)=C_{t,k}(M)$, 
then it is highly probable that the position set $\tilde{{\cal P}}_{n,t,k}$ is 
equal to the correct position set ${\cal P}_{n,t,k}$. 
The attacker repeats this process for each $k\in \mathbb{Z}[1,2^t-1]$.
\end{enumerate}

\subsection{Value of Position Checker Determines the Position Exponent}
The basic question is whether the observed values of the 
position checker $C_{t,\ell}(M)$ for $M$ 
determine the position set ${\cal P}_{n,t,\ell}$ uniquely 
in attack process 3. 

Generally, the answer is no. 
It is possible that ${\cal \tilde{P}}_{n,t,\ell}\ne {\cal P}_{n,t,\ell}$ 
even though 
$M^{d[t,{\cal \tilde{P}}_{n,t,\ell}]}\equiv C_{t,\ell}(M)(~\Mod N)$ holds for some $M$.
For example, $C_{t,\ell}(1)=1$ for 
every ${\cal P}_{n,t,\ell}$, and $C_{t,\ell}(-1)=1$ for every even $d[t,{\cal P}_{n,t,\ell}]$.
In this case, the attacker cannot determine the position set ${\cal P}_{n,t,\ell}$ correctly using 
the value of $C_{t,\ell}(M)$ for such $M$.
This depends on the order of $M$ mod $N$, which is defined as 
the least positive integer $s$ that satisfies $M^s\equiv 1(\Mod N)$, denoted by $\ord_N(M)$. 
There are a few $M$ of small order because the number of 
roots of $x^{\ord_N(M)}\equiv 1(\Mod N)$ is at most $\ord_N(M)^2$. 
Here the exponent ``$2$" is from the number of factors of $N$. 
Carmichael's theorem states that $M^{\lambda(N)}\equiv 1(\Mod N)$ for every $M\in\mathbb{Z}_N^*$ and 
the order of $M$ mod $N$ is a factor of $\lambda(N)=\lcm(p-1,q-1)$. 

Suppose the attacker obtains a non-trivial pair of exponents such that 
$M^{s_1}\equiv M^{s_2}(\Mod N)$ with $s_1>s_2$. 
Then, $s_1-s_2$ is a factor of $\lambda(N)$. 
In many case, $s_1-s_2$ is large because there are a few $M$ of small order. 
Hence, $s_1-s_2$ gives the attacker information that is useful for factorizing 
the public modulus $N$.
Consequently, we expect that such a situation will not occur frequently.

As will be discussed in Section \ref{SIM}, 
from many situations, we did not obtained an incorrect 
estimate for the position exponent from 
a value of the position checker. 
Thus, it seems that its error probability 
is negligible. 
Therefore, the position checker works practically for most values of $M$.

\section{Attack Simulation}\label{SIM}

\subsection{Simulation Settings and Overview}
Computer simulations were performed based on the attack procedure 
presented in Section \ref{attackproc}. 
The simulation steps are as follows, where steps 1 and 2 are 
to be performed on the user side or by the coprocessor, 
and steps 3 - 5 are the attacker's process. 
Note that, from simple power analysis and supporting information, 
the attacker is assumed to know the value of $t$.
\begin{enumerate}
\item A secret exponent $d = e^{-1} ~\Mod (p-1)(q-1)$ 
is calculated after randomly generating two distinct prime numbers $p$ and $q$ of 
$n$ bits each. 
A message $M$ and a public exponent $e$ are also randomly generated.
\item The faulted signatures $\hat{S}_{t,k}$ are generated for 
$k\in\mathbb{Z}[1,2^t-1]$. 
Here, let the correct signature be denoted 
$\hat{S}_{t,2^t}$. 
These signatures are known by the attacker.
\item The attacker creates the position checkers $C_{t,k}(M)=\hat{S}_{t,k+1}\hat{S}_{t,k}^{-1}~\Mod N$ 
from faulted signatures $\hat{S}_{t,k}$ and $\hat{S}_{t,k+1}$ for $k\in\mathbb{Z}[1,2^t-1]$.
\item To determine $d$, $n$ bits of $d$ are divided for every $t$ bits 
from the least significant bit (LSB) to the most significant bit (MSB). 
A $t$-bit chunk is called a ``block" and is numbered 1 through 
$B=\lceil n/t\rceil$ from a block that includes the MSB to a block 
that includes the LSB.
\item The $d$-search algorithm described in Section \ref{d-search} is then performed. 
\end{enumerate}

\subsection{$d$-search algorithm}\label{d-search}

An example of the 
$d$-search algorithm is presented in LIST III, 
where a coefficient in the $b$-th block of $d$ is stored in $\mbox{de}[b]$.

\begin{center}
{\small LIST III} \\
{\small Search Process Example}
\end{center}
{\footnotesize
\begin{verbatim}
a)  // initialize
de[b] = -1 for b=1,2,c,B; 
coef[k] = -1 for k = 1,2,c,2^t-1
b)  // find coefficient coef[k] which appear once 
    // at de[ja] in B blocks.
  for (ja = 1; ja <=B; ja++) {
    if (de[ja] == -1) {
      expA = (2^t)^(B-ja);
      posId = powmod(M, expA, N);
      If (posId ¸ C_k(M)) {
        find k such that posID = C_k(M);
        de[ja] = k; coef[k] = 1;
      }
    }
  }
c)  // find coefficient coef[k] which appear 
    // two times at {de[ja], de[jb]} in B blocks.
  for (ja = 1; ja<=B; ja++) {
    if (de[ja] == -1) {
      expA = (2^t)^(B-ja);
      for (jb = ja+1; jb<=B; jb++) {
        if (de[jb] == -1 {
          expB = expA + (2^t)^(B-jb);
          posId = powmod(M, expB, N);
          If (posId ¸ C_k(M)) {
            find k such that posID = C_k(M);
            de[ja] = de[jb] = k; coef[k] = 2;
            goto nxt;
          }
        }
      }
    }
    label nxt;
  }
d)  find coefficient coef[k] which appear up to 
    Lmt (<=B) times at {de[ja], de[jb], .., de[Lmt]} 
    for k = 1,2,c,2^t-1
e)  // the remaining positions are for coefficient 0.
if b satisfies de[b] = -1, then let de[b] = 0;
\end{verbatim}
}

Here $\mbox{coef}[k]$ is a search flag. 
When the position of coefficient $k$ is found, $\mbox{coef}[k]$ becomes $1$, 
otherwise it becomes $-1$.
The $\mbox{de}[b]$ for $b\in\mathbb{Z}[0,2^t-1]$ is 
the secret exponent at the end of the search process.
The background of this algorithm is as follows.
The probability of one coefficient having $z$ blocks in $d$ is given by
$$
p(B,z) = {B \choose
z}\left(\frac{1}{w}\right)^z\left(1-\frac{1}{w}\right)^{B-z},
$$
where $w=2^t$.
Let the expected number of coefficients having $z$ positions each be
$\overline{W}_z$,
and the expected number of occupied positions by all coefficients having $z$ blocks
each is $\overline{B}_z$.
These expectations are given by $\overline{W}_z=wp(B,z)$ and
$\overline{B}_z=zw p(B,z)$.
If the coefficient $0$ is excluded for these expectations, they become 
$\overline{W}_z'=(w-1)p(B,z)$ and $\overline{B}_z'=z(w-1)p(B,z)$.
On the other hand, $z$ blocks can be chosen from $B$ blocks by ${B \choose z}$.
Therefore, the probability to pass the position check for a
combination of $z$ blocks in $B$
blocks is given by
$$
\overline{W}_z'/{B \choose z}=\frac{(w-1)^{B-z+1}}{w^B}
$$
which decreases when $z$ becomes large.
This means that we should first check positions for $z=1$ for quick results.
Then, we can determine the average $\overline{B}_1'$ positions. 
The coefficients having two positions each should be searched,
and the probability to pass the position check is
$$
\frac{(w-1)^{B-\overline{B}_1'-1}}{w^{B-\overline{B}_1'}}.
$$
In the same manner, next search is for coefficients having an 
incremental number of positions are then searched.

\subsection{Simulation Results}
The experiment was performed on a computer with a 16-core(XENON) 3.8 GHz CPU and 256 GB of memory.
We implemented the attack presented in LIST III using 
the Number Theory Library \cite{NTL} developed by Victor Shoup.
In this experiment, we measured the processing time required 
to derive the secret exponent $d$.
Attack simulation was performed for $(n,t) = (1536, 6)$ $101$ times, 
where $d$, $e$, and $M$ were randomly generated $101$ times. 
All estimation results for each $d$ were confirmed to be correct. 
The median and average CPU times required for the attack were approximately 
962 and 1120 min, respectively. 
The attack time for the worst case was 5189 min, and the best attack time was 249 min.

\section{Small $t$ case}
Smaller $t$ makes our attack difficult because 
the size of the position set ${\cal P}_{n,t,\ell}$ for most $\ell\in \mathbb{Z}[0,2^t-1]$ 
is too large to search by brute force.
Under such an unfavourable attack condition 
the attacker can determine the position of several ``rare" coefficients of the secret exponent $d$. 
However, for a small $t$ case, the attacker generally requires significant time 
to determine the complete secret exponent $d$ even for $t=5$ and $n=1024$.

On the other hand, for a small $t$ case, the attacker can reconstruct the complete $d$ 
if the public exponent $e$ is relatively small. 
A small public exponent $e$ (e.g., $65537=2^{16}+1$) is 
often used to speed up signature verification.
By observing the RSA equation 
$$
ed-k\varphi(N)=1,
$$
we have 
$$
d = \frac{k}{e}\varphi(N)+\frac{1}{e}=\frac{k}{e}(N-s)+\frac{1}{e},
$$
where $s=p+q-1$. Suppose $p$ and $q$ are balanced, i.e., $p<q<2p$. 
Then, we have $s=p+q-1<3\sqrt{N}$.
Note that $k$ takes a value that satisfies $\gcd(e,k)=1$ and $k\in\mathbb{Z}[1,e-1]$ 
because $d<\varphi(N)$. 
The length of $s$ is approximately half that of $N$. 
Therefore, it is highly probability that 
the upper bits of $d$ will be equal to the upper bits of $kN/e$ when $e$ is small.
If the attacker knows some partial coefficients in the DCA process, 
they can determine $k$ uniquely by comparing the upper bits of $kN/e$ with those coefficients. 
Thus, the attacker can determine the upper bits of $d$ after determining a unique $k$. 
In this manner, the DCA process can be accelerated.
For example, for $e=65537$, the attacker can obtain the upper half bits of $d$ by 
$e-1=65536$ trials at most. 
The attacker must only search the coefficients using the proposed 
attack method for most of the lower half bits of $d$.
For the case $n=1024$ and $t=5$, 
we expect that more than half the coefficients can be searched in a short time.
Our simulation shows that all coefficients of $d$ can be searched 
within a few days by Mathematica 9 on a Windows 7 PC (2.3 GHz).

\section{Software countermeasures against DCA}

\subsection{Small $t$ and large $e$}
A simple and effective countermeasure against DCA 
is to adopt small $t$ and $e$ of length that is approximately equal to the length of $N$. 
For instance, consider the case where $t=4$ and $N$ is 1536-bits long. 
In this case, 
the number of positions for the same coefficient $\leq 20$ is approximately $2^{109.9767}$. 
In this case, it is difficult for the proposed DCA alone to find the complete secret exponent.

In addition, this countermeasure reduces RAM space for the table.
Clearly, RAM space is reduced by half, and, 
as can be seen in Table \ref{averagetime}, 
the average computation time increases by approximately 3.2\% 
when $t$ is reduced from $6$ to $4$. 

If a chip developer adopts this countermeasure, 
the choice of $t$, $e$, and the length of $N$ must be considered very carefully.
For example, the attacker can determine the secret exponent using DCA in a short time 
for a 512-bit $N$ even if $t=4$. 

The attacker can determine several coefficients using the naive attack 
described in Section \ref{Naive} 
if the chip does not provide other countermeasures against fault attacks.

\subsection{Execution Redundancy}
Execution redundancy is the repeatition of computations and 
comparison of results to verify that the 
correct result is generated~\cite{Bar}. 
The simplest execution redundancy is 
recomputation, i.e., performing a computation twice on the same hardware. 
This countermeasure executes the same computation twice and compares the first and second results. 
If these results are not equal, then a fault is detected. 
The recomputation consists of precomputation and exponentiation processes.
Generally, in order to accelerate of recomputation, only exponentiation process
executes twice.
However, our attack works if the chip executes precomputation only once.
Therefore, precomputing twice are indispensable to prevent our attack.

Another effective countermeasure is an inverse computation, i.e., 
the encryption and comparison of $M$ to an encrypted signature~\cite{Anderson}. 
This is not only effective but also is faster than the recomputation 
countermeasure if the public exponent is relatively small such as 65537.
A condition by which an inverse computation works successfully is given in Theorem \ref{workcond}.
\begin{thm}\label{workcond}
Pick $M\in\mathbb{Z}_N^*$. 
Suppose $\ord(M)$ is not a common divisor 
of $\varphi(N)$ and $d[t,\cup_{k\leq \ell\leq 2^t-1}{\cal P}_{n,t,\ell}]$ 
for every $k\in\mathbb{Z}[1,2^t-1]$.
Then, for some $k$, the inverse computation detects the faulted signatures $\hat{S}_{t,k}$ 
created from faulted table ${\cal \hat{T}}_{2^t}(k)$ with message $M$. 
\end{thm}
\begin{proof}
As discussed in Section \ref{genposck}, the faulted signature 
created from faulted table ${\cal \hat{T}}_{2^t}(k)$ with message $M$ is represented as
$$
\hat{S}_{t,k} = M^{d-d[t,\cup_{k\leq \ell\leq 2^t-1}{\cal P}_{n,t,\ell}]}~\Mod N.
$$
Then, $\hat{S}_{t,k}^e$ is given by
\begin{eqnarray*}
\hat{S}_{t,k}^e &\equiv& 
M^{ed-ed[t,\cup_{k\leq \ell\leq 2^t-1}{\cal P}_{n,t,\ell}]} \\
&\equiv& M\cdot M^{-ed[t,\cup_{k\leq \ell\leq 2^t-1}{\cal P}_{n,t,\ell}]}(\Mod N)
\end{eqnarray*}
If the inverse computation does not detect this fault, i.e., $\hat{S}_{t,k}^e \equiv M$ holds, then,
$$
M^{ed[t,\cup_{k\leq \ell\leq 2^t-1}{\cal P}_{n,t,\ell}]}\equiv 1(\Mod N).
$$
Thus, $\ord(M)$ is a common divisor of $ed[t,\cup_{k\leq \ell\leq 2^t-1}{\cal P}_{n,t,\ell}]$ for every $k\in\mathbb{Z}[1,2^t-1]$. 
Since $\ord(M)$ is a factor of $\varphi(N)$ and $\gcd(e,\varphi(N))=1$, 
$\ord(M)$ is a common divisor 
of $\varphi(N)$ and $d[t,\cup_{k\leq \ell\leq 2^t-1}{\cal P}_{n,t,\ell}]$ 
for every $k\in\mathbb{Z}[1,2^t-1]$, which gives the desired result.
\end{proof}

\subsection{Exponent randomization}
The exponent randomization method was proposed by P. Kocher~\cite{timingattack} 
in 1996 to defeat power analysis.
The exponent randomization picks a random integer $r$ 
and computes a digital signature with 
$\tilde{d}=d+r\varphi(N)$ rather than the raw $d$. 
Coron~\cite{coron} extended this countermeasure to an elliptic curve cryptosystem.
For an RSA case, this is an effective countermeasure against numerous fault attacks. 
For example, the safe error attack developed by Yen-Joye~\cite{SafeError}
does not work if the card developer adopts exponent randomization 
because the attack uses bit-flipping of the secret exponent. 
Exponent randomization also works for other fault attacks [\cite{Bao}, \cite{Bar}] 
by relying on the stationary nature of the secret exponent.
Berzati et al.~\cite{Berzati} developed an effective fault attack against 
exponent randomization. 
Their attack succeeds from approximately 1000 
faulty signatures for a 1024-bit RSA signature algorithm.
Nevertheless, exponent randomization makes it difficult for 
their fault attacks to reconstruct the complete secret exponent.
On the other hand, precomputation is independent of the exponent; 
therefore, this countermeasure has no adverse effect on the proposed DCA.

\section{Conclusion}
In this paper, we have introduced a position checker to determine 
the position set ${\cal P}_{n,t,{\ell}}$ for 
$2^t$-ary modular exponentiation and have proposed a double counting attack 
that uses the proposed position checker against a $2^t$-ary RSA signature. 
The attack is the first fault attack based on an instruction skip 
against precomputation. 
We have also provided an estimate of the probability of 
uniquely determining the position exponent from 
a value obtained by the position checker. 
We also performed attack simulations for several parameters, including 1536-bit RSA 
with the $2^6$-ary method. 
In addition, if the public exponent is small, search time is reduced. 
From several attack simulations for public exponent $e=65537$, we have also shown that 
1024-bit RSA with the $2^5$-ary method can be broken in a short time.

Empirically, we did not obtain any incorrect estimates 
for the position exponent from a value of the 
position checker in many simulations. 
Thus, it appears that the error probability of the proposed method 
is significantly small. 
Moreover, we have also discussed efficacy of several software countermeasures. 
We have concluded that neither 
recomputation with a single precomputation nor the use of a randomized exponent 
have adverse effect against our attack.

\section*{acknowledgements}
This work was supported by JSPS KAKENHI Grant Number 25330157.


\begin{thebibliography}{}
%
%
\bibitem{Bellcore} D. Boneh, R.A. DeMillo, and R.J. Lipton, ``On the importance
of eliminating errors in cryptographic computations," J. Cryptol. vol.14, no.2,
pp. 101-119, Springer-Verlag, Berlin, 2001. Earlier version published in EUROCRYPT' 97.
%
\bibitem{LowCost} R. Anderson and M. Kuhn, ``Low cost attacks on tamper resistant devices," 
In IWSP: 5th International Workshop on Security Procotols, LNCS 1361, 
Springer-Verlag, pp. 125-136, 1998.
%
\bibitem{Anderson}
R. Anderson and S. Skorobogatov, ``Optical fault injection attacks," 
In Workshop on Cryptographic Hardware and Embedded Systems(CHES 2002), LNCS 2523, Springer-Verlag, pp. 2-12, 2002.
%
\bibitem{Bar} H. Bar-Ei, H. Choukri, D. Naccache, M. Tunstall, and C. Whelan, 
``The sorcerer's apprentice guide to fault attacks," In Proc. of the IEEE 94(2), pp. 370-382, 2006.
%
%
\bibitem{Joye-Tunstall-Book} M. Joye, and M. Tunstall eds., 
{\it Fault Analysis in Cryptography}, Springer, 2012.
%
%
\bibitem{PKCS} RSA Labs., PKCS\#1: RSA CRYPTOGRAPHY STANDARD, 
http://japan.emc.com/emc-plus/rsa-labs/standards-initiatives/pkcs-rsa-cryptography-standard.htm
%
\bibitem{Lenstra}A. K. Lenstra, ``Memo on RSA signature generation in the presence of faults," 1996.
%
\bibitem{Bao} F. Bao, R. H. Deng, Y. Han, A. B. Jeng, A. D. Narasimharu, and T. H. Ngair, 
``Breaking public key cryptosystems on tamper resistant devices in the presence of transient faults," 
Lecture Notes in Computer Science vol. 1361, pp. 115-124, 1998.
%

\bibitem{SafeError} S. M. Yen, and M. Joye, ``Checking before output may not be 
enough against fault-based cryptanalysis," IEEE Trans. Comput. 49(9), pp.967-970, 2000.
%
\bibitem{Seifert} J. Seifert, 
``On authenticated computing and RSA-based authentication." In Proc. of the 12th
ACM Conference on Computer and Communications Security (CCS 2005), November 2005, pp. 122-127.
%
\bibitem{Muir} J. A. Muir, ``Seifert's RSA fault attack: Simplified analysis and generalizations, "
In ICICS'06 Proc. of the 8th International Conference on Information and Communications Security, 
pp.420-434.
%
\bibitem{SmartcardHandbook}W. Rankl, and W. Effing, {\it Smart Card Handbook 4$^{th}$ed.}, 
John Wiley \& Sons, 2010.
%
\bibitem{DPA}
P. C. Kocher, J. Jaffe, and B. Jun, ``Differential Power Analysis." In 
M. J. Wiener (ed.): Advances in Cryptology.CRYPTO'99, 19th
Annual International Cryptology Conference, Santa Barbara, California,
USA, August 15.19, 1999, Proc. Lecture Notes 
in Computer Science, vol. 1666, pp. 388-397. Springer, Berlin, 1999.
%
\bibitem{Mont} P. Montgomery, 
``Modular Multiplication Without Trial Division," Math. Computation, vol. 44, pp. 519--521, 1985.
%
\bibitem{HandbookofAppliedCryptography}
A. J. Menezes, P. C. van Oorschot, and S. A. Vanstone,
{\it Handbook of Applied Cryptography}, CRC Press, 1996.
%
\bibitem{RoundReduction} H. Choukri and M. Tunstall,
``Round reduction using faults," In Proc. of FDTC, pp.13-24, 2005.
%
\bibitem{park} J. Park, S. Moon, D. Choi, Y. Kang, and J. Ha,~
``Differential fault analysis for round-reduced AES by fault injection,"~
ETRI J. vol. 33, no.3, pp.434-441, 2011.
%
\bibitem{yoshikawa1} H. Yoshikawa, M. Kaminaga, and A. Shikoda,~
``Round addition using faults for generalized Feistel network,"~
{\em IEICE Trans. Inf. \& Syst.}, vol.E96-D, no.1, pp.146-150, 2013.
%
\bibitem{KaminagaDFA}
M. Kaminaga, T. Watanabe, T. Endo, and T. Okochi, 
``Logic-level analysis of fault attacks and a cost-effective countermeasure design," 
IEICE Transactions E91-A(7), pp. 1816-1819, 2008.
%
\bibitem{NTL}
V. Shoup, Number Theory C++ Library(NTL) version 6.1.0. Available at www.shoup.net
%
\bibitem{timingattack} P. C. Kocher, ``Timing Attacks on Implementations of Diffie-Hellman, RSA, DSS, and
Other Systems." In N. Koblitz (ed.): Advances in Cryptology (CRYPTO 1996),
Lecture Notes in Computer Science vol. 1109, pp. 104-113. Springer, 1996.
%
\bibitem{coron} J-S. Coron, ``Resistance Against Differential Power Analysis for Elliptic Curve
Cryptosystems." In C. K. Koc and C. Paar (ed.): Cryptographic Hardware and
Embedded Systems (CHES 1999), Lecture Notes in Computer Science vol. 1717,
pp. 292-302. Springer, 1999.
%
\bibitem{Berzati} A. Berzati, C. Canovas-Dumas, and L. Goubin, 
``Public key perturbation of randomized RSA implementations," 
Cryptographic Hardware and Embedded Systems (CHES 2010), LNCS vol. 6225, pp. 306-319, 2010.
%
%
%
%
%
%
%
%
\end{thebibliography}
\end{document}